\pdfoutput=1
\documentclass[conference]{IEEEtran}
\IEEEoverridecommandlockouts
\PassOptionsToPackage{hyphens}{url} 
\usepackage{amsmath,amssymb,amsthm}
\usepackage{graphicx}
\usepackage{booktabs}
\usepackage{courier}   
\usepackage{enumitem}
\setlist{nosep,leftmargin=1.2em}
\usepackage[hidelinks]{hyperref}

\newcommand{\Xbar}{\bar{X}}
\newcommand{\Acal}{\mathcal{A}}
\newcommand{\Tcal}{\mathcal{T}}
\newcommand{\Ecal}{\mathcal{E}}
\newcommand{\Rn}{\mathcal{R}_{\mathrm{nat}}}
\newcommand{\ivr}{\textsc{IVR}}
\newtheorem{proposition}{Proposition}
\newtheorem{observation}{Observation}

\begin{document}
\title{What You Sign Is Not What You Run:\\ The Verifiability Gap in Quantized LLM Supply Chains}
\author{\IEEEauthorblockN{Wei-Liang Liu}\IEEEauthorblockA{National Yang Ming Chiao Tung University\\ rlooooong@gmail.com}
\and
\IEEEauthorblockN{Cheng-Yen Tsai}\IEEEauthorblockA{National Taiwan University of Science and Technology\\ nemo00407@gmail.com}}
\maketitle

\begin{abstract}
Model signing authenticates the full-precision weights $W$ an LLM is released with, yet users run a \emph{quantized} artifact $Q$ whose producing transformation is not uniquely determined under current disclosure. We separate the non-environment inputs from the execution environment $E$ and measure a quantity that, to our knowledge, has not been measured before: the \emph{Independent Verifiability Rate} (\ivr{}), the fraction of a popularity-sampled GGUF ecosystem that a third party can re-derive bit-for-bit. Across the 200 most-downloaded Hugging Face GGUF repositories (198 in scope), no artifact or model card records the build environment, so under a recorded-inputs protocol the \ivr{} is $0/198$. Yet the artifacts are not irreproducible. On four in-sample repositories that disclose their inputs, the card's information alone reproduces no \texttt{Q4\_K\_M} file, but two coordinates the card omits---the build's floating-point-contraction setting and, for hybrid models, the per-tensor quantization recipe---re-derive $7/8$ target--quant files bit-for-bit on x86 ($5/8$ on arm64); the fp16 re-derives from the provider's checkpoint, but the calibration imatrix must be shipped and hash-bound, since recomputing it did not match. An older release additionally pins the instruction set and did not re-derive. Reproducing the divergence reported with \texttt{llama.cpp} PR~\#25353, we show that $E$ is a necessary coordinate and that one compiler setting yields a unique output across the tested platforms at negligible quality cost ($+0.06\%$ perplexity). As a security characterization, benign floating-point non-determinism separates from uniformly placed nibble-flip tampering above a detection floor that grows with the environments a verifier must allow, but a targeted embedding edit evades the statistic, so the proposed protocol rests on exact re-derivation; adversarial-fp16 attacks lie outside it. We release the audit data and a re-runnable harness.
\end{abstract}

\section{Introduction}\label{sec:intro}
Reproducible builds turned ``trust the publisher'' into ``verify the artifact''~\cite{lamb2022reproducible}. 
A decade of engineering made over 96\% of Debian's packages (bookworm/amd64) \emph{reproducible} in its continuous rebuild tests~\cite{reproduciblebuilds}: building a
package twice under deliberately varied environments yields the \emph{same bytes}, so an independent rebuild can be
checked against the released artifact. Source authentication thereby becomes an end-to-end integrity
check. ML supply chains are adopting the same primitives---model signing 
(Sigstore Model Signing~\cite{sigstoremodeltransparency,openssfmodelsigning}, NVIDIA NGC~\cite{nvidia2025ngcsigning}) 
authenticates the hashes of the released files, and lifecycle-attestation efforts (the IETF 
\texttt{ai-model-lifecycle-attestation} draft~\cite{sharif2026ietfattest}, CycloneDX 
ML-BOM~\cite{cyclonedxmlbom}) record provenance.

There is a gap these primitives do not close. Users of local-inference ecosystems typically do not run the signed full-precision
weights $W$; they run a \emph{quantized} artifact $Q$---e.g.\ a 4-bit GGUF widely distributed through 
\texttt{llama.cpp}~\cite{llamacpp}/\texttt{ollama} ecosystems. $Q$ is produced by a transformation 
$\Tcal$ whose inputs go far beyond $W$: the quantization recipe $\theta$, the calibration input
$C$ (an \emph{importance matrix}, ``imatrix'': per-weight calibration statistics computed by running
the model over a corpus), the execution environment $E$ (quantizer release and build, compiler, architecture, floating-point settings),
and non-weight components $M$ (tokenizer, chat template, metadata). Writing the non-environment 
inputs as $\Xbar=(W,\theta,C,M)$, we observe that under currently disclosed specifications $Q$ is 
not uniquely determined: signing $H(W)$ says nothing about which $Q$ a user should accept, and 
re-running $\Tcal$ on a different platform can produce different bytes.

We ask a question that, to our knowledge, the ML supply chain has not measured:
\emph{of the GGUF repositories in our popularity-sampled ecosystem, what fraction expose an artifact that a third party can independently re-derive bit-for-bit?} We call this the Independent Verifiability Rate,
\[
\ivr{} = \frac{\#\,\{\text{bit-exact re-derivable artifacts}\}}{\#\,\{\text{evaluated artifacts}\}}.
\]
We fix \ivr{} to bit-exact re-derivability; attestation-based binding (e.g.\ by a Trusted Execution
Environment, TEE: hardware that signs a statement of which binary ran on which inputs) is a separate
mechanism (\S\ref{sec:protocol}), not part of this metric.

\paragraph{Our answer, under a recorded-inputs protocol, is $0/198$} None of the 198 in-scope repositories among the 200 most-downloaded exposes enough to satisfy our re-derivation protocol. That means the disclosed inputs are insufficient for \emph{our} verifier, not that the artifacts are mathematically unreproducible. Part of this is structural: GGUF has no field for the build environment, so the metadata half of our checklist cannot pass, and the census measures what the format and current disclosure practice allow; the model-card scan and a re-derivation attempt test whether disclosure elsewhere closes the gap. A verifier can also succeed by \emph{searching} environments instead of reading a record; the trial of \S\ref{sec:m1} shows this works for recent releases on the publisher's architecture and fails for an older one, so the search-based rate is not zero but is unmeasured beyond our trial. We then explain \emph{why}, \emph{what it costs}, and \emph{how to close it}, as three measurements feeding a verification protocol (Figs.~\ref{fig:ivr}--\ref{fig:security}):
\begin{itemize}
\item \textbf{M1 (\S\ref{sec:m1}):} the recorded-inputs \ivr{} is $0/198$; no artifact or model card records the build environment---compiler, architecture, floating-point settings (19/198 cards do state the \texttt{llama.cpp} release). On four in-sample repositories, re-deriving the deployed files from the card's disclosed inputs (release, publisher's own fp16, imatrix, stock preset) reproduces none of the \texttt{Q4\_K\_M} files. Supplying two coordinates the card omits---the build's FP-contraction setting and, for hybrid models, the per-tensor recipe---re-derives $7/8$ target--quant files bit-for-bit on x86, and for a pilot repository the fp16 itself re-derives bit-for-bit from the provider's checkpoint.
\item \textbf{M2 (\S\ref{sec:m2}):} reproducing the floating-point-contraction divergence reported with PR~\#25353, we find the execution environment $E$ empirically necessary under the studied model; one compiler setting attains $B{=}0$ for the tested types across the tested architectures and compilers at negligible measured quality cost ($\Delta$PPL $+0.06\%$), yet is not the \texttt{llama.cpp} default, and the divergence lives in the weight representation, not inert metadata.
\item \textbf{M3 (\S\ref{sec:m3}):} a security \emph{characterization}---our evaluated statistics separate benign floating-point non-determinism from uniformly placed nibble-flip tampering above a detection floor that grows with the range of build environments a verifier must allow (up to about $0.2\%$ of weights under our five-pair calibration, and more as that range widens); a targeted embedding edit evades them, and adversarial-fp16 attacks lie outside re-derivation.
\item \textbf{Protocol (\S\ref{sec:protocol}):} An ambiguity-free transformation ($B{=}0$) is the natural precondition for trustless re-derivation; uniform partial spot-checking against an $m$-unit tampering within $N$ units needs $\Theta(N/m)$ openings to hold a fixed detection probability.
\end{itemize}

\paragraph{Contributions}
(i) To our knowledge, the first popularity-sampled measurement of GGUF independent verifiability, decomposed into what is \emph{obtainable} (source-reconstructible, $148/198$) versus what is \emph{recorded} (the build environment: $0/198$, a structural consequence of the format having no field for it). The constructive core of the paper is a trial re-derivation on named public artifacts (five repositories across three model families and three releases, \texttt{b3772}--\texttt{b10603}): it locates the failure in the quantization step and names the card-absent coordinates that close it. With those supplied, $7/8$ target--quant files of the four Qwen repositories and both files of the Granite one re-derive bit-for-bit, and for a pilot repository the fp16 already re-derives from the provider's checkpoint.
(ii) A compact vocabulary---$\Xbar/E$, the admissible set $\Acal$, and a $0/1$ ambiguity indicator $B$---that makes precise what a complete re-derivation specification must fix; under it the necessity of recording the environment $E$ is immediate, and we witness it empirically (\S\ref{sec:m2}). We use this as scaffolding for the measurements, not as a result in itself.
(iii) A security characterization of the $B{>}0$ regime: how much one-step tampering hides in benign environment variation (up to about $0.2\%$ of weights under our five-pair calibration, growing with the range of environments allowed), what a max statistic adds against sparse large edits, and a delineated threat boundary. In our sweeps, a uniform scale perturbation is caught before it is harmful and adaptive placement hides only a bounded low-magnitude reservoir, but a \emph{targeted} embedding attack evades the statistic (\S\ref{sec:m3}); this is why the protocol relies on exact re-derivation rather than detection.
Along the way we independently reproduce the floating-point-contraction divergence of PR~\#25353, add a weight-space localization and a size trend, and confirm its perplexity report with a paired estimate.
We also provide a re-runnable audit harness and our analysis tools (GGUF header-probe, byte-region diff, and detector), distinct from the external re-derivation engine we instrument.

\paragraph{Prior art we build on, not claim} Deterministic GGUF re-derivation is existing engineering (ggufpacker, which pins and packages quantization inputs, publishes cross-machine bit-exact derivation attestations for quants built with the fix, and reports that for quants built before it \texttt{Q8\_0} packs are exact while most k-quants need a correction delta and cannot be attested~\cite{ggufpacker}; llama.cpp PR~\#25353~\cite{llamacppfpcontract}), as is the observation that k-quant outputs diverge across platforms because of floating-point contraction. The adversarial-fp16 attack that bounds our threat model is also prior work (Mind the Gap~\cite{egashira2025mindthegap}). Our contribution is to \emph{measure} how far current disclosure falls short of what these tools presuppose, \emph{name} what a manifest must bind for re-derivation to succeed, and \emph{delineate} what re-derivation can and cannot cover. The trial re-derivation (\S\ref{sec:m1}, Table~\ref{tab:trial}) also makes explicit a coordinate that a packager pins silently: on named public artifacts, the card's quant type (``\texttt{Q4\_K\_M}'') underdetermines the publisher's \emph{per-tensor recipe}, a discrete input that the card does not disclose and that no stock preset regenerates; it must be supplied alongside the environment before re-derivation succeeds.

\section{Background \& Threat Model}\label{sec:bg}
\paragraph{Reproducible builds} The reproducible-builds movement made the integrity of a running
binary independently checkable: recompile from source in a recorded environment and confirm
bit-for-bit that it matches the shipped binary, supported by tooling that locates non-determinism
(\texttt{diffoscope}~\cite{diffoscope}, \texttt{reprotest}~\cite{reprotest}) and a network of
independent rebuilders. Much of the residual non-determinism it confronts is \emph{normalizable
metadata}---embedded timestamps, build paths, filesystem ordering, locale---which is canonicalized
away; for the rest of the build environment, the practice is to either make it irrelevant to the
output or record it so the artifact can be rebuilt exactly (e.g.\ Debian's \texttt{.buildinfo}
files)~\cite{rbrecording}. A central empirical question of this paper is how quantization fares under
this recipe. Normalization cannot help---the non-determinism lies in the numeric content of the
weights (\S\ref{sec:m2})---so the environment must be made irrelevant (one compiler setting achieves this
on the tested platforms, \S\ref{sec:m2}) or recorded, which no artifact in our sample fully does (\S\ref{sec:m1}).

\paragraph{Model signing \& attestation} The ML ecosystem is adopting the signing half of this picture. Sigstore Model Signing~\cite{sigstoremodeltransparency,openssfmodelsigning} and NVIDIA NGC~\cite{nvidia2025ngcsigning} sign a manifest of the released files' hashes (the weights $W$ and whatever else the repository ships), binding an artifact to a publisher's identity. Provenance and lifecycle efforts go further in describing the pipeline: the IETF \texttt{ai-model-lifecycle-attestation} draft (an individual Internet-Draft)~\cite{sharif2026ietfattest} records a source-weight hash, a quantized-artifact hash, a quantization-tool hash, and per-layer error tolerances (a check of the kind a sub-threshold targeted edit evades, \S\ref{sec:m3}); CycloneDX ML-BOM enumerates model components~\cite{cyclonedxmlbom}; and work on attesting LLM pipelines gates artifact promotion on verifiable training and release claims~\cite{tan2026attesting}. None of the specifications we examine requires binding the calibration input $C$ (the imatrix) as an input of the quantization (a file-manifest signature covers an imatrix the repository happens to ship, but does not tie it to the artifact derived from it), and $E$ is bound at most indirectly: a hash of the quantization tool identifies the quantizer binary but does not describe it, so a verifier on another platform cannot re-create an equivalent build. These claims establish provenance and tool identity, but do not by themselves let an independent verifier establish that the artifact is the faithful output of the signed source.

\paragraph{The transformation} A deployed GGUF is $Q=\Tcal(\Xbar;E)$, where $\Xbar=(W,\theta,C,M)$ collects the weights, recipe, calibration input, and non-weight components, and $E$ is the environment. The calibration input is itself a pipeline output, $C=\mathcal{C}(W,D,E_C)$ over a calibration corpus $D$ and its own environment $E_C$, which is why binding $C$ by hash (as we recommend) matters. Deterministic re-derivation of GGUF outputs is existing engineering---ggufpacker~\cite{ggufpacker} packages the inputs, and \texttt{llama.cpp}~\cite{llamacpp} PR \#25353~\cite{llamacppfpcontract} removes a floating-point non-determinism---and we use these as an \emph{instrument} to measure the boundary they presuppose, not as a contribution we claim. Throughout, $W$ denotes the full-precision artifact actually fed to $\Tcal$---for GGUF, the publisher's released fp16 GGUF; the upstream conversion of the provider's checkpoint to fp16 is a separate transformation stage, which we test once (\S\ref{sec:m1}) and otherwise leave out of scope.

\paragraph{Block quantization in one paragraph} A GGUF k-quant~\cite{kquants} stores weights in fixed-size blocks, grouped into 256-weight \emph{superblocks}: each block holds low-bit integer codes for its weights---4-bit codes are packed two per byte as \emph{nibbles}---plus per-block \emph{scale} (and minimum) parameters, so that a dequantized weight is approximately $\text{scale}\times\text{code}$ and the spacing between adjacent representable values (one \emph{quantization step}) is set by the scale. The scales are not read off directly: a small floating-point optimization (the \emph{scale search}) picks the scale minimizing the block's reconstruction error, weighted by the imatrix when one is supplied. Two consequences carry the rest of the paper: the scale search is where floating-point environment differences can change the chosen scales and hence the stored bytes (\S\ref{sec:m2}), and altering a stored code moves the corresponding weight by at least one step, whereas perturbing a scale moves every weight in its block by sub-step amounts (\S\ref{sec:m3}).

\paragraph{Threat model} In the signed supply chain that model-signing efforts target, the model provider signs $H(W)$, while a
quantizer---usually a \emph{different}, third party---publishes $Q$ and implicitly asserts $Q=\Tcal(\Xbar;E)$;
the current ecosystem does not generally provide an independently verifiable cryptographic binding 
between the signed source and the transformation result. The recommended protocol we analyze 
(\S\ref{sec:protocol}) addresses this gap by having the transforming party sign a manifest 
committing to the inputs it used, $(H(W),\theta,C,E,M)$, so that any third party can recompute and 
check. The manifest binds $Q$ to the inputs the transforming party commits to, and for GGUF that party usually supplies three of them itself: $W$ (its own fp16 conversion), $C$ (its own imatrix), and $M$. Re-derivation therefore certifies the fp16$\to Q$ step. The conversion that produced $W$ and $M$ can be re-derived too---for the pilot repository of \S\ref{sec:m1} we reproduce the published fp16 bit-for-bit from the model provider's checkpoint---which extends the chain to the provider's weights; a committed $C$, however, shows which imatrix was used, not that it is benign. Our security claim is regime-split (Table~\ref{tab:regimes}); Appendix~\ref{app:assume} states the assumptions A1--A7 and the exact boundary of the claim.

\begin{table}[t]\centering\small
\caption{Security claim by regime. A \emph{cheating quantizer} starts from a benign $W$ but ships $Q\neq\Tcal(\Xbar;E)$; $B_{\Ecal'}{=}0$ means the honest output is unique over the environments rebuilders use (\S\ref{sec:formal}).}\label{tab:regimes}
\begin{tabular}{@{}p{0.29\linewidth}p{0.64\linewidth}@{}}
\toprule
Threat and regime & What re-derivation gives a verifier \\
\midrule
Cheating quantizer, $B_{\Ecal'}{=}0$ & Sound detection: a hash mismatch implies tampering, with no false positives (A2, A4). \\
Cheating quantizer, $B_{\Ecal'}{>}0$ & A nonzero diff no longer implies tampering; uniformly placed nibble flips are only \emph{characterized}, by a measured detection floor, and a targeted edit can evade the statistic (\S\ref{sec:m3}). \\
Cheating quantizer via $C$ (a chosen imatrix, $Q{=}\Tcal$) & Nothing beyond binding: the manifest fixes which $C$ was used and re-derivation reproduces $Q$; vetting $C$ is outside re-derivation (\S\ref{sec:bg}). \\
Malicious $W$ (adversarial fp16, $Q{=}\Tcal$) & Nothing: the quantization is honest, so re-derivation reproduces the artifact (out of scope, A1). \\
\bottomrule
\end{tabular}
\end{table}

\paragraph{Why not simply sign $Q$} A natural objection is that the quantizer could simply sign $H(Q)$. This proves \emph{provenance}---that $Q$ came from whoever holds the signing key---but not \emph{faithfulness} to the signed source $W$: a cheating quantizer signs its own backdoored $Q$ just as validly, and, because the quantizer is usually not the party that signed $W$, no signature links the two. Binding $Q$ to $W$ requires \emph{recomputing} $Q$ from $W$ and comparing hashes over the verifier's own recomputation. This is why $B{=}0$ matters: when two honest recomputations of the same inputs can diverge, there is no canonical $H(Q)$ for a verifier to compare against, and a signature over one particular $Q$ does not, by itself, establish faithfulness to $W$.

\section{Transformation Ambiguity}\label{sec:formal}
Both the measurement (\S\ref{sec:m1}) and its explanation (\S\ref{sec:m2}) need one piece of vocabulary, which we fix before turning to the data. Split non-environment inputs from the environment: $\Xbar=(W,\theta,C,M)$, full input $X=(\Xbar,E)$,
with $E$ ranging over $\Ecal$, the space of admissible execution environments. The admissible-output set is $\Acal(\Xbar)=\{\Tcal(\Xbar;E):E\in\Ecal\}$ and the
ambiguity budget $B(\Xbar)=\sup_{Q_1,Q_2\in\Acal(\Xbar)} d(Q_1,Q_2)$, where $d$ is the discrete metric on serialized artifacts ($d{=}1$ iff the bytes differ), so
$|\Acal(\Xbar)|{=}1 \Leftrightarrow B{=}0$---exactly the uniqueness of $H(Q)$ that hash comparison requires; under this metric $B$ is simply a 0/1 indicator of non-uniqueness. The same supremum over a subset $\Ecal'\subseteq\Ecal$ defines $B_{\Ecal'}(\Xbar)$: disclosing a coordinate of $E$ shrinks the environments a verifier must consider, so what verification needs is $B_{\Ecal'}{=}0$ over the environments rebuilders actually use---met by exactly replicating a fully recorded $E$, or, portably, by a setting under which all of them agree. When we say a setting ``attains $B{=}0$'' we mean $B_{\Ecal'}{=}0$ for $\Ecal'$ the tested platforms with that setting; Observation~\ref{prop:e} concerns the full $\Ecal$. (Where the \emph{size} of a divergence matters rather than its existence, \S\ref{sec:m3} measures perturbations in dequantized weight space.) We write $X^*$ for a \emph{minimal complete specification}---a smallest set of disclosed coordinates making the output unique ($|\Acal|{=}1$); Observation~\ref{prop:e} shows every complete specification must constrain $E$, which we abbreviate $E\in X^*$; characterizing $X^*$ in full is left open. The \emph{natural-variance distribution} $\Rn$ denotes the empirically
characterized distribution of weight perturbations induced by benign variation in $E$. An observation (Observation~\ref{prop:e}, necessity of $E$,
\S\ref{sec:m2}) and a standard sampling bound (Proposition~\ref{prop:sample}, \S\ref{sec:protocol}) are
stated where they do work; proofs are in Appendix~\ref{app:proofs}.

\section{M1: The Independent Verifiability Rate}\label{sec:m1}
We operationalize the question of \S\ref{sec:intro} as a measurement over the artifacts people actually download.

\paragraph{Definitions} An artifact $Q$ is \emph{bit-exact verifiable} if an independent party, given only the disclosed inputs, can recompute a $Q'$ with $H(Q')=H(Q)$; the \ivr{} is the popularity-sampled fraction of deployed artifacts that are. Verifiability is a property of \emph{disclosure}, not of the artifact in principle: a $Q$ that could in fact be reproduced, were its producer to reveal $(\Xbar,E)$, still counts as non-verifiable while that disclosure is absent. Our census scores a \emph{recorded-inputs} protocol: all four inputs, including the environment, must be disclosed, so that a mismatch is attributable. A verifier may instead \emph{search} plausible environments and accept a hash match; that route needs no record, and the trial below bounds it from below. To separate ``the inputs are hidden'' from ``the environment is hidden,'' we additionally report the weaker \emph{source-reconstructible} rate---whether the weights $W$ and calibration input $C$ are retrievable at all, ignoring $E$. Throughout, for GGUF we take $W$ to be the publisher's released fp16 GGUF (the full-precision artifact that is fed to quantization). For the source-reconstructible rate we report two rules: whether the repository ships such a GGUF, and whether its declared base-model repository holds public full-precision weights. The transformation we re-derive and verify is the quantization step $\text{fp16}\!\to\!Q$, not the upstream conversion of an original framework checkpoint to fp16, which is a separate stage outside our scope.

\paragraph{Sampling frame} Verifiability is not uniform over a randomly chosen model; what matters is what users run, so we sample by popularity. We take the 200 most-downloaded GGUF repositories returned by the Hugging Face API (August 2026, sorted by the API's download count, a trailing 30-day window); two expose no GGUF weight file and fail inclusion, which fixes the denominator at 198 repositories (from 61 distinct publishers), each scored on one representative deployed artifact (below). Because download rankings drift, we freeze the exact repository list and ship it, with the scored snapshot, in the artifact. The scoring is a re-runnable audit (\texttt{census.py -{}-repos-file}) that re-scores the frozen list against live metadata, so that a reviewer can check the $0/198$ rather than trust it (repositories change over time, so a re-run may drift from the snapshot).

\paragraph{Instrument} Downloading 198 multi-gigabyte repositories is unnecessary: what decides verifiability lives in the GGUF header and the repository's advertised provenance. For each repo we query the HF API for the file inventory, select one representative deployed artifact (preferring its \texttt{Q4\_K\_M} file; the first shard for sharded models), and fetch only that GGUF's header via HTTP range requests ($8$--$16$\,MB rather than the full file) to parse its key--value block. We then score a four-input re-derivation checklist: (i) is the fp16/source $W$ retrievable; (ii) is the quantizer and its \emph{version} recorded; (iii) is the calibration imatrix $C$ present and retrievable (tri-state: an artifact with no imatrix record is treated as $C=\varnothing$, i.e.\ fully disclosed rather than a gap---an assumption that misclassifies an imatrix whose record was removed with \texttt{--hide-imatrix}); (iv) is any execution-environment coordinate $E$ recorded. No standardized GGUF field carries either the quantizer version or an environment coordinate (\texttt{general.quantization\_version} records the revision of the quantization format, not the \texttt{llama.cpp} build), so from metadata checks (ii) and (iv) are scored jointly as a single environment/version record (the ``env/version'' bar of Fig.~\ref{fig:ivr}b); the model-card analysis below separates them. Bit-exact verifiability requires all four; source-reconstructibility requires only (i) and (iii). For the studied quantization step, these checks cover the inputs our re-derivation procedure consumes. The recipe $\theta$ is the quant type named in the artifact (its per-tensor types are recorded in the header) plus options that default per release, the release being part of $E$ in our model. The non-weight components $M$ (tokenizer, chat template, metadata) are taken verbatim from the published GGUF container, whose metadata region an environment rebuild reproduces bit-identically apart from the local file paths that \texttt{llama-quantize} records for an imatrix (below and \S\ref{sec:m2}). The quantizer does write \emph{into} that block: $62/198$ artifacts carry \texttt{quantize.imatrix.*} provenance keys---which is also how the census detects imatrix use, alongside the repository file listing---but these record $C$-side provenance and no $E$ coordinate. We treat $C$ as an \emph{obtained} input (retrievable or not); verifying its own generation ($C=\mathcal{C}(W,D,E_C)$, \S\ref{sec:bg}) is a separate stage outside this measurement.

\paragraph{Result: $\mathbf{\ivr{}=0/198}$ under the recorded-inputs protocol} No repository in the frozen set satisfies the recorded-inputs protocol. The source-reconstructible rate is far higher---$148/198$ ($74.7\%$) under the census's rule, which credits public full-precision weights in the declared base-model repository (an upper bound: for a fine-tuned derivative that is the parent model), or $111/198$ ($56.1\%$) if we require the repository itself to ship a full-precision GGUF. For roughly half to three quarters of these models, then, the weights and calibration data are in principle obtainable; what is missing lies elsewhere. The universal gap is the execution environment: no standardized GGUF field records the \texttt{llama.cpp} version, compiler, target architecture, or FP-contraction setting, and $0/198$ artifacts record any $E$ coordinate in metadata (19 cards state the release in prose, below). Table~\ref{tab:census} decomposes the result (Fig.~\ref{fig:ivr}b plots the first half). Beyond the environment/version record that all $198$ lack, $50$ repositories also fail source reconstruction; the remaining $148$ are missing \emph{only} that record. The environment is thus not one gap among several but the universal one, by construction of the format. For the $154$ artifacts whose representative file is a ggml scale-search k- or i-quant, the mechanism of \S\ref{sec:m2} is expected to apply---every scale-search type we tested diverged under default builds, though some of these files come from other ggml-based tools we did not test---so the missing record is not a bookkeeping omission: a rebuild in an unrecorded environment can byte-differ from the honest output, and a mismatch cannot then be attributed (the trial below shows both outcomes). For the other $44$ (23 \texttt{Q8\_0}/\texttt{Q4\_0} files---\texttt{Q8\_0} matched across our platforms at \texttt{b10448}, while \texttt{Q4\_0} quantized with an imatrix uses a weighted scale search we did not test; 8 f16 or vision-projector files our selector picked where no quantized language-model file was preferred; and 13 custom or mixed-precision formats), the necessity of $E$ is not established, and $\ivr{}=0$ reflects the protocol's requirement of a recorded release and environment rather than a demonstrated barrier. Restricting to text and multimodal LLMs leaves the result unchanged ($0/165$). Read the other way, the search-based rate is not zero: in the trial below, three in-sample \texttt{Q4\_K\_M} files re-derive bit-for-bit from disclosed inputs alone on the publisher's architecture, so at least $3/198$ representative files are re-derivable by a verifier willing to guess the environment; measuring that rate across the ecosystem requires re-deriving every obtainable repository, which we leave open.

\begin{table}[t]\centering\small
\caption{Census decomposition over the 198 repositories. Restricting to text and multimodal LLMs by Hugging Face pipeline tag (excluding speech, video, image-generation, embedding, and classification models) leaves $\ivr{}=0/165$.}\label{tab:census}
\begin{tabular}{@{}lr@{}}
\toprule
\multicolumn{2}{@{}l}{\emph{What the repositories lack}}\\
Environment/version record & 198\\
Retrievable fp16 source (base-model rule) & 33\\
Retrievable imatrix (2 also lack fp16) & 19\\
Only the environment/version record & 148\\
\midrule
\multicolumn{2}{@{}l}{\emph{Representative file scored}}\\
Scale-search k/i-quant & 154\\
\texttt{Q8\_0}/\texttt{Q4\_0} (no scale search without an imatrix) & 23\\
f16 or vision projector & 8\\
Custom or mixed-precision format & 13\\
\bottomrule
\end{tabular}
\end{table}

\paragraph{Model cards} A metadata checklist could miss provenance stated in prose (model-card completeness on Hugging Face is known to be uneven~\cite{liang2024modelcards}), so we checked the human-written model cards in two ways. First, an adversarial audit of 12 cards, selected from the top-60 census stratum as the nine highest-download repositories whose only gap is the environment/version record plus three with deeper gaps. One independent LLM agent per card was instructed to \emph{refute} the ``not verifiable'' label by finding a disclosed exact quantizer version, build flags, or environment coordinate (a card saying merely ``quantized with \texttt{llama.cpp}'' does not count); the agent saw the scorer's label, which may anchor it, whereas the deterministic scan below does not depend on it. It confirmed $12/12$ (on source availability it disagreed with the scorer for three repositories, two in one direction and one in the other, which we report but do not apply); none of these 12 cards states an exact \texttt{llama.cpp} version, one states partial build flags, and several advertise a \emph{custom} quantizer. Second, because a 12-card sample can miss publisher-specific conventions, we pattern-scanned all $198$ cards for release and commit identifiers and manually classified every hit: $19/198$ cards---18 of them from two publishers' standard templates, and none from a publisher in the audited sample---state the \texttt{llama.cpp} release used for quantization. The scan finds no card stating a compiler, target architecture, or FP-contraction setting (60 cards mention platforms such as Apple Silicon, Metal, or AVX; on manual review all refer to inference or performance, none to the quantization build). The quantizer release is thus sometimes disclosed in prose; beyond one card's partial build flags, the build environment is disclosed nowhere we look. We scope the finding precisely: the build coordinates of $E$ are absent from the GGUF metadata and the model cards of all $198$ artifacts---``absent in the locations we audit,'' not a claim that the producer never possessed it.

\paragraph{Source availability is not verifiability---and what closes the gap} The gap is not merely bureaucratic, and these artifacts are not irreproducible: what is disclosed is insufficient, and we can say exactly what is missing. We take four in-sample \texttt{bartowski} repositories---\texttt{Qwen2.5-7B-Instruct} (release \texttt{b3772}) and \texttt{Qwen3.5-\{0.8B,2B,4B\}} (release \texttt{b9222})---each of which publishes its own full-precision GGUF, its imatrix, and \texttt{Q4\_K\_M}/\texttt{Q8\_0} files, and whose card names the \texttt{llama.cpp} release used. From exactly those disclosed inputs we re-derive the deployed files, comparing per-tensor \texttt{sha256} of the tensor data (Table~\ref{tab:trial}); the key--value region is excluded because \texttt{llama-quantize} writes the local imatrix path into it (above), which a whole-file hash would need canonicalized. With only the card's information---the named release, the published imatrix, the stock preset---no \texttt{Q4\_K\_M} file re-derives, and only a minority of its tensors match ($140/335$ and $141/339$ for the two families); the one file among these eight that does re-derive from the card alone is the dense 7B's \texttt{Q8\_0}, whose stock preset equals the publisher's recipe. The gap decomposes into two coordinates the card never records, neither of which is the weight data. \textbf{(1)~The per-tensor recipe:} for the hybrid Qwen3.5 models the stock \texttt{Q4\_K\_M} preset does not reproduce the publisher's type map, which keeps the SSM tensors at F32 and bumps others to \texttt{Q5\_K}/\texttt{Q6\_K}/\texttt{Q8\_0} ($75$--$100$ tensors differ), though the card records only ``\texttt{Q4\_K\_M}''. Pinning the published types (a few hundred discrete labels) closes this; for the dense 7B the stock preset already matches, and no recipe is needed. The recipe gap spans the publisher's whole lineup: for the 0.8B, re-deriving \emph{all 23} published quant types---k-quants, legacy quants, i-quants, and the embedding/output-override variants---reproduces $0/23$ from the card's named type alone but $23/23$ bit-for-bit once the recipe is supplied (on x86, whose default build is uncontracted). The pattern is not specific to this family or release. A model from a different provider and family, \texttt{granite-4.2-3b} (IBM Granite, a distinct architecture, quantized by the same publisher at the newer release \texttt{b10603}), reproduces its \texttt{Q4\_K\_M} in only $263/363$ tensors from the card alone but $363/363$ bit-for-bit with the recipe; there the publisher bumps the attention tensors to \texttt{Q5\_K}/\texttt{Q6\_K}, a different hidden recipe from Qwen3.5's F32-kept SSM tensors. Its \texttt{Q8\_0} reproduces from the card alone. \textbf{(2)~The build's floating-point contraction}, which is architecture-conditioned: with the recipe pinned, the Qwen3.5 \texttt{Q4\_K\_M} re-derives \emph{bit-for-bit} on x86 (default \emph{and} \texttt{-ffp-contract=off}) and on arm64 with \texttt{-ffp-contract=off}, but not on an arm64 default build. At \texttt{b9222} the quant kernels are compiled portably, so baseline x86-64 (no FMA) is uncontracted by default, while arm64 (NEON FMA, contracted by clang's default) needs PR~\#25353; once contraction is off the \texttt{Q4\_K\_M} output is bit-identical across the two architectures. With both coordinates supplied, $7/8$ of the target$\times$quant files re-derive bit-for-bit on x86, and $5/8$ on arm64: the 0.8B and 2B \texttt{Q8\_0} keep one residual tensor there whose all-zero blocks store byte-different but numerically identical codes (the arm64 and x86 kernels round $0/0$ differently), a residual that contraction does not touch---bit-exactness is strictly stronger than numerical equality, and a cross-architecture rebuilder must record the architecture or canonicalize zero blocks. The exception is the 7B \texttt{Q4\_K\_M}: its release \texttt{b3772} compiles the quant kernels with \texttt{-march=native}, so the environment additionally pins the SIMD ISA, and it matches none of our six configurations (two compilers on arm64 and one on x86, each with contraction on and off; best $170/339$). We did not find the publisher's environment; a different compiler version or build host may add coordinates beyond the ISA. A pilot on a further repository from the same publisher (below the census's download cutoff), \texttt{Qwen2.5-1.5B-Instruct} (also \texttt{b3772}), behaved the same way: its four rebuilds (two platforms, two builds, with its published imatrix) gave three distinct honest outputs, none matching the published file. The conversion hop, by contrast, re-derives exactly: for that 1.5B repository, converting the model provider's checkpoint (safetensors and model card, Hugging Face revision \texttt{989aa79}) with the named converter reproduces the published fp16 bit-for-bit, on a different machine and twice over (the newer converter at \texttt{b10448} reproduces the tensor data but writes different metadata); we verified this hop on that one repository only. Given those disclosed inputs, the failure to verify sits in the quantization step, and is localized to a few coordinates the model card does not disclose: for a recent release, the build's contraction setting (for which the GGUF container has no field) and the per-tensor recipe (present in the deployed file's header, but not regenerated by the card's named preset), both small and disclosable. An older release additionally pins the instruction set, and, as the next paragraph shows, the calibration input must be obtained and hash-bound rather than recomputed.

\begin{table}[t]\centering\footnotesize\setlength{\tabcolsep}{3pt}
\caption{Trial re-derivation: per-tensor bit-exact matches (tensor count in parentheses). \emph{card}=named release + published imatrix + stock preset; \emph{+rec}=card plus the deployed file's per-tensor type map. On x86 the default build is already uncontracted, so its two builds coincide. $^*$=fully bit-exact; ---=the stock preset already equals the recipe; $^\dagger$=one residual tensor whose all-zero blocks store byte-different, numerically identical codes; $^\ddagger$=\texttt{b3772} builds the kernels \texttt{-march=native} (SIMD ISA pinned; the arm64/gcc-15 builds, not shown, equal arm64/clang; x86 default $170$, contraction off $141$); n/t=not tested; $\neq$=tensor-data hash differs in all four rebuilds (three distinct honest outputs). No model card records either coordinate.}
\label{tab:trial}
\begin{tabular}{@{}llccc@{}}
\toprule
 & & arm64 default & arm64 contr.\ off & x86 (either) \\
target (tensors) & quant & card\,/\,+rec & card\,/\,+rec & card\,/\,+rec \\
\midrule
Qwen3.5-0.8B (335) & \texttt{Q4\_K\_M} & 140\,/\,193 & 260\,/\,\textbf{335}$^*$ & 260\,/\,\textbf{335}$^*$ \\
Qwen3.5-2B (335)   & \texttt{Q4\_K\_M} & 140\,/\,193 & 260\,/\,\textbf{335}$^*$ & 260\,/\,\textbf{335}$^*$ \\
Qwen3.5-4B (441)   & \texttt{Q4\_K\_M} & 184\,/\,252 & 341\,/\,\textbf{441}$^*$ & 341\,/\,\textbf{441}$^*$ \\
Qwen3.5-0.8B (335) & \texttt{Q8\_0}   & 298\,/\,334$^\dagger$ & 298\,/\,334$^\dagger$ & 299\,/\,\textbf{335}$^*$ \\
Qwen3.5-2B (335)   & \texttt{Q8\_0}   & 298\,/\,334$^\dagger$ & 298\,/\,334$^\dagger$ & 299\,/\,\textbf{335}$^*$ \\
Qwen3.5-4B (441)   & \texttt{Q8\_0}   & 393\,/\,\textbf{441}$^*$ & 393\,/\,\textbf{441}$^*$ & 393\,/\,\textbf{441}$^*$ \\
Qwen2.5-7B (339)   & \texttt{Q8\_0}   & \textbf{339}$^*$\,/\,--- & \textbf{339}$^*$\,/\,--- & \textbf{339}$^*$\,/\,--- \\
Qwen2.5-7B (339)   & \texttt{Q4\_K\_M} & 141\,/\,---$^\ddagger$ & 141\,/\,---$^\ddagger$ & 170\,/\,---$^\ddagger$ \\
granite-4.2-3b (363) & \texttt{Q4\_K\_M} & n/t & n/t & 263\,/\,\textbf{363}$^*$ \\
granite-4.2-3b (363) & \texttt{Q8\_0}   & n/t & n/t & \textbf{363}$^*$\,/\,--- \\
Qwen2.5-1.5B pilot & \texttt{Q4\_K\_M} & $\neq$\,/\,---$^\ddagger$ & $\neq$\,/\,---$^\ddagger$ & $\neq$\,/\,---$^\ddagger$ \\
\bottomrule
\end{tabular}
\end{table}

\paragraph{The calibration input is a further unverifiable hop} Re-derivation above used the publisher's \emph{own} imatrix, taken as given. Can a verifier instead recompute it? We regenerate the 0.8B imatrix from the fp16 and the calibration corpus its card links (identical \texttt{sha256}) at the stated release. The structure re-derives exactly---the entry set and all $186$ entries' token \emph{counts} match---but the accumulated importance values do not: $0/186$ are bit-identical, differing by a median relative $5.5\times10^{-5}$ from the published imatrix on both an arm64 and an x86 CPU (which themselves differ by $2.2\times10^{-5}$), while \texttt{-ffp-contract=off} leaves the arm64 imatrix unchanged. The sensitivity is thus to the architecture rather than to contraction, and since no CPU build we tested matches the published imatrix, it was presumably produced on an accumulation path we did not test. The difference is not inert: re-quantizing with our recomputed imatrix instead of the published one---at the recipe and \texttt{-ffp-contract=off} environment that otherwise reproduce the file bit-for-bit---changes $141/335$ tensors. A verifier must therefore \emph{obtain} the exact imatrix and bind it by hash, as our protocol does (\S\ref{sec:protocol}); recomputing $C$ from the same disclosed corpus yields a different $C$, and hence a different $Q$. Calibration-input generation is thus a third hop in the chain, alongside the fp16 conversion and the quantization itself, and the second we find to be environment-sensitive.

\begin{figure}[t]\centering
\includegraphics[width=\linewidth]{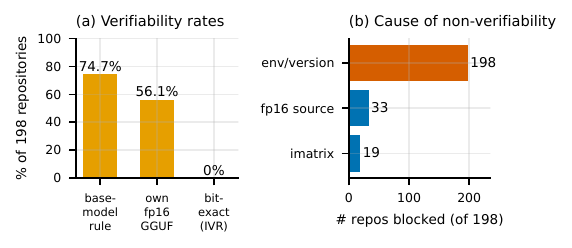}
\caption{(a) Source-reconstructible rates under the two rules (both ignore $E$) and the Independent Verifiability Rate under the recorded-inputs protocol; (b) the causes of non-verifiability. Both over the 198 GGUF-bearing repositories of the top 200, from metadata and model cards.}\label{fig:ivr}
\end{figure}

\section{M2: Locating the Ambiguity}\label{sec:m2}
If $0/198$ artifacts record $E$, the natural rejoinder is that $E$ may not matter---that running the same recipe on the same weights yields the same bytes regardless of environment, as a deterministic pure function would. We test this directly. We build \texttt{llama.cpp} at the pinned release build \texttt{b10448} on two platforms---arm64/clang (Apple~M4) and x86\_64/gcc (Linux), the two (architecture, compiler) pairs tested throughout---each in a \emph{default} and a \texttt{-ffp-contract=off} variant (PR~\#25353). On each we quantize one fixed fp16 GGUF (the f16 published by \texttt{bartowski/Qwen2.5-1.5B-Instruct-GGUF}, the pilot repository of \S\ref{sec:m1}; \texttt{sha256} identical on both machines; no imatrix, $C{=}\varnothing$) into four types: the scale-search k/i-quants \texttt{Q4\_K\_M}, \texttt{Q5\_K\_M}, and \texttt{IQ4\_XS}, plus \texttt{Q8\_0} as a control that carries no scale search. (Table~\ref{tab:setup} lists every experimental setup.) Same-binary, same-machine quantization is bit-deterministic ($16/16$ configurations---2 platforms $\times$ 2 builds $\times$ 4 types, each re-run 10 times with a single \texttt{sha256}), as expected; that determinism is the baseline, not the finding. The cross-platform divergence below, its floating-point-contraction cause, and the flag that removes it were reported alongside ggufpacker and PR~\#25353~\cite{ggufpacker,llamacppfpcontract}; we reproduce them independently, as the empirical witness for Observation~\ref{prop:e}, and add a weight-space localization, a size trend, and a paired quality check.

\begin{table}[t]\centering\footnotesize
\caption{Experimental setups. Models are Qwen2.5-Instruct unless stated; builds are CPU-only \texttt{llama.cpp} in a default and a \texttt{-ffp-contract=off} variant on arm64/clang (Apple~M4) and x86\_64/gcc (Linux). At \texttt{b10448} the ggml quant kernels compile without \texttt{-march=native}, so the x86 default build's quant path is baseline x86-64 (no FMA, hence uncontracted) while arm64 (NEON, FMA available) contracts by clang's default; this is why \texttt{-ffp-contract=off} is the setting that makes the two architectures agree for the scale-search types (the \texttt{Q8\_0} all-zero-block residual at \texttt{b9222}, \S\ref{sec:m1}, is a separate architecture dependence). A \texttt{\_M} preset stores its k-quant tensors in the named format, so \S\ref{sec:m3} writes \texttt{Q4\_K} for \texttt{Q4\_K\_M}.}\label{tab:setup}
\begin{tabular}{@{}p{0.26\linewidth}p{0.15\linewidth}p{0.48\linewidth}@{}}
\toprule
Experiment & Release & Inputs and conditions\\
\midrule
Cross-architecture (\S\ref{sec:m2}) & \texttt{b10448} & 1.5B publisher f16; \texttt{Q4\_K\_M}, \texttt{Q5\_K\_M}, \texttt{IQ4\_XS}, \texttt{Q8\_0}; no imatrix; 10 repeats\\
Size trend (\S\ref{sec:m2}) & \texttt{b10448} & 1.5B, 7B, 14B; \texttt{Q4\_K\_M}; arm64 default vs.\ contraction off\\
Quality (\S\ref{sec:m2}) & \texttt{b10448} & 1.5B; \texttt{Q4\_K\_M}; wikitext-2, 584 chunks\\
Pilot re-derivation (\S\ref{sec:m1}) & \texttt{b3772} & 1.5B publisher f16 and imatrix; \texttt{Q4\_K\_M}; both platforms and builds\\
Trial re-derivation (\S\ref{sec:m1}) & \texttt{b3772}, \texttt{b9222}, \texttt{b10603} & 7B, Qwen3.5-0.8B/2B/4B, granite-4.2-3b; publisher full-precision GGUF and imatrix; \texttt{Q4\_K\_M}, \texttt{Q8\_0} (all 23 types for the 0.8B); stock preset vs.\ pinned recipe; x86\_64 and arm64, default and contraction off\\
Imatrix regeneration (\S\ref{sec:m1}) & \texttt{b9222} & Qwen3.5-0.8B; card-linked corpus; arm64 (both builds) and x86\_64\\
Conversion (\S\ref{sec:m1}) & \texttt{b3772}, \texttt{b10448} & 1.5B provider safetensors to f16; x86\_64; 2 runs each\\
Security (\S\ref{sec:m3}) & \texttt{b10448}, \texttt{b3772} & 1.5B (0.5B for one overlap); \texttt{Q4\_K}, \texttt{Q5\_K}, \texttt{Q6\_K}, \texttt{IQ4\_XS}; six honest pairs for \texttt{Q4\_K}; 5 seeds\\
Adaptive tampering (\S\ref{sec:m3}) & \texttt{b10448} & 1.5B; \texttt{Q4\_K} scale $\times(1{+}\epsilon)$ sweep with perplexity; adaptive placement; targeted \texttt{Q6\_K} embedding edit, two probes\\
\bottomrule
\end{tabular}
\end{table}

\paragraph{Quantization is not a pure function of $(W,\theta)$} With $\Xbar$ held fixed, the two architectures' \emph{default} builds produce byte-differing outputs for $3$ of the $4$ quant types; only \texttt{Q8\_0}, which carries no k-quant scale search, matches across architectures ($1/4$). As PR~\#25353 reports, the divergence traces to compiler-controlled floating-point contraction---fusing $a\cdot b+c$ into a single fused multiply-add---inside the k-quant scale-search loop of \texttt{ggml-quants.c}: the fused and unfused forms round differently near ties, select different block scales, and thereby write different nibbles; our byte localization below is consistent with this. Compiling with \texttt{-ffp-contract=off} collapses all four types to bit-identical across both tested (architecture, compiler) pairs ($4/4$ types), i.e.\ $B_{\Ecal'}(\Xbar){=}0$ for $\Ecal'$ the two tested pairs with contraction disabled. The same pattern holds at an older release and with a real calibration input ($C{\neq}\varnothing$): in the \texttt{b3772} pilot re-derivation of \S\ref{sec:m1} (1.5B), the two \texttt{-ffp-contract=off} outputs are tensor-identical across platforms with and without the imatrix, while the default outputs differ. The environment therefore silently changes the \emph{weights}, and a single compiler flag removes the observed cross-environment divergence. The flag does not, however, \emph{eliminate} $E$: it is itself a recorded $E$ coordinate, and the quantizer version still varies, so it makes recording $E$ cheap---a version plus one build flag---rather than unnecessary (\S\ref{sec:protocol}). The switch also has negligible quality cost in the configuration we evaluated (Qwen2.5-1.5B-Instruct, \texttt{Q4\_K\_M}, wikitext-2~\cite{merity2017pointer}): with each build quantizing the same f16, the \texttt{-ffp-contract=off} artifact scores \emph{perplexity} (the standard next-token quality metric; lower is better) $10.750$ versus $10.744$ for the default build---$\Delta\text{PPL}=+0.006$ ($+0.06\%$; paired standard error $0.005$ over the $584$ evaluation chunks, so not distinguishable from zero)---though the two artifacts byte-differ, in line with PR~\#25353's report of effectively unchanged perplexity. We did not evaluate other models, types, or tasks.

\paragraph{The divergence lives in the weights, not in metadata} Byte-localizing the cross-architecture diff ($C{=}\varnothing$), we find the GGUF key--value/metadata region bit-identical; all divergence falls in the quantized weight tensors, concentrated in the k-quant per-block scale parameters (per-byte difference rate $1.2\times10^{-3}$, the highest of any region) and propagating into the packed nibbles (which hold $73\%$ of the differing bytes by count but at a lower per-byte rate, $3.4\times10^{-4}$). This is precisely where the normalization half of the reproducible-builds recipe fails to transfer: beyond the local file paths recorded with an imatrix (which could be canonicalized), the non-determinism is in the numeric content of the weights themselves, so metadata normalization cannot reach it. That leaves the recipe's other two remedies: making $E$ irrelevant (as the flag does) or recording it (\S\ref{sec:bg}).

\paragraph{The effect persists with scale} The same compiler-controlled divergence appears at 1.5B, 7B, and 14B parameters (Qwen2.5-Instruct f16 at each size), and across the three sizes we measured the fraction of divergent bytes grows with model size (Fig.~\ref{fig:repro}b). The cross-architecture comparison is measured \emph{directly} at 1.5B; at 7B and 14B we report a single-machine proxy: arm64 default vs.\ arm64 \texttt{-ffp-contract=off}. The proxy is exact at 1.5B: at \texttt{b10448} the x86\_64 default output is byte-identical to every \texttt{-ffp-contract=off} output, so the proxy compares the same two byte streams as the cross-architecture pair and the two diffs coincide bit-for-bit.

\begin{observation}[Necessity of $E$, under the studied model]\label{prop:e}
For the \texttt{llama.cpp} k-quant family $\{\Tcal(\Xbar;E):E\in\Ecal\}$, there exists an instance $\Xbar$ with $|\Acal(\Xbar)|\ge2$, hence $B(\Xbar)>0$; so a specification scheme that pins $\Xbar$ but leaves $E$ free is not complete, and any complete specification must constrain $E$, i.e.\ $E\in X^*$.
\end{observation}
\emph{Proof sketch.} Exhibit one instance ($\theta{=}\texttt{Q4\_K\_M}$) and two environments (x86\_64/gcc vs arm64/clang, defaults) whose honest outputs byte-differ (\S\ref{sec:m2}). Appendix~\ref{app:proofs}; the experiments separately identify $\Xbar\cup\{E_{\text{fpcontract}}\}$ as a \emph{sufficient} completion for the tested environments. We claim empirical necessity \emph{for this model}, not that every model requires $E$, and we do not claim to have characterized $X^*$ in full.

\textbf{Scope (where CUDA enters).} \texttt{llama.cpp} quantization runs on CPU, so a CUDA-enabled build executes the same CPU quantization code (we did not test CUDA builds separately). GPU matters for (i) imatrix generation $C$, which \S\ref{sec:m1} already finds architecture-sensitive on CPU and which GPU reduction order may perturb further (GPU builds untested here), and (ii) GPU-native pipelines (GPTQ/AWQ), which may lack $B{=}0$ and fall to attestation.

\begin{figure}[t]\centering
\includegraphics[width=\linewidth]{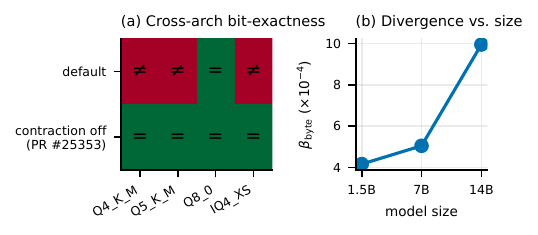}
\caption{(a) Cross-architecture bit-exactness at 1.5B: default builds byte-differ for 3 of 4 types; \texttt{-ffp-contract=off} restores bit-identity ($4/4$). (b) The fraction of divergent bytes $\beta_{\mathrm{byte}}$ (arm64 default vs \texttt{-ffp-contract=off}, the single-machine proxy of \S\ref{sec:m2}; \texttt{Q4\_K\_M}) grows with model size.}\label{fig:repro}
\end{figure}

\section{M3: Security Characterization}\label{sec:m3}
This section is a security \emph{characterization}, not a proposed defense. The question it answers is whether the ambiguity of \S\ref{sec:m2} is dangerous: when $B_{\Ecal'}{>}0$, an honest cross-environment rebuild already yields a nonzero diff, so does benign non-determinism \emph{look like} tampering and leave a verifier unable to tell them apart? We measure the two perturbation distributions in dequantized weight space. For each weight tensor, $\Delta w$ denotes the per-weight difference between the two dequantized artifacts compared, and $\sigma$ the standard deviation of that tensor's dequantized weights; all perturbation magnitudes below are reported as the normalized $|\Delta w|/\sigma$.

\paragraph{Two signatures} Benign non-determinism---the natural-variance distribution $\Rn$---is measured here as the divergence of the arm64 default rebuild from a \emph{reference} output (the \texttt{-ffp-contract=off} build, which both tested platforms reproduce bit-for-bit, \S\ref{sec:m2}; a designated measurement reference, not a canonicity claim), i.e.\ floating-point-contraction variance at a fixed release (\texttt{b10448}) without an imatrix. It touches only a small fraction of weights ($0.22\%$ for the \texttt{Q4\_K\_M} artifact, $0.056\%$ for \texttt{Q5\_K\_M}; Qwen2.5-1.5B-Instruct) and moves each by a \emph{tiny} amount: median $|\Delta w|/\sigma = 0.014$ (99th percentile $0.26$) and $0.0017$ respectively, because the scale drift nudges the affected weights by a sub-step. A cheating quantizer that tampers by flipping nibbles ($Q\neq\Tcal$) moves each touched weight by one quantization step of its block---median $|\Delta w|/\sigma \approx 0.27$ for uniformly placed Q4\_K flips, roughly $20\times$ larger (finer formats have smaller steps: $0.13$ for Q5\_K, $0.06$ for Q6\_K). The per-weight change is thus bounded below by the local step size (\S\ref{sec:bg}), not by how few weights are touched: it cannot be shrunk by making the attack sparser. It can be shrunk by \emph{placement}: in blocks with small scales one step is small ($0.43\%$ of the Q4\_K weights have a one-step move below the threshold used below), so an adaptive attacker confined to such blocks could evade it; the floors below are for uniformly placed flips, and adaptive placement is measured at the end of this section. We validated the per-flip step size against the attack's arithmetic: a synthetic $\pm1$-step flip gives median $0.268$, matching a direct diagnostic on real weights ($2{,}000$ nibble flips applied in place to the 0.5B model: median $0.277$).

\paragraph{Separation and detection floor} The two magnitude distributions barely overlap: the overlap coefficient (histogram intersection, $\sum_b \min(p^{\mathrm{nat}}_b,p^{\mathrm{tamper}}_b)$ over a fixed $|\Delta w|/\sigma$ binning, which compares the magnitudes of changed weights, not their number) is $0.036$ for one-step nibble flips on Q4\_K at 0.5B (Qwen2.5-0.5B-Instruct), and $0.076$ (Q4\_K) and $0.053$ (Q5\_K) at 1.5B. The operative quantity is the \emph{detection floor}: the smallest tampering budget still distinguishable from benign non-determinism. Because the separation is one of magnitude, a simple statistic realizes it---a per-tensor tail fraction, the fraction of weights with $|\Delta w|/\sigma$ above $\tau{=}0.1$. We build each tampered artifact the way a cheating quantizer on another platform would produce it---its honest rebuild plus uniformly placed one-step flips---and score it against the verifier's reference (five seeds; budgets from $0.002\%$ of weights). Against the arm64 contraction pair at a known release, the per-tensor floor (the smallest budget with AUC $\ge0.99$ in every seed, over $165$--$197$ tensors per type) is $0.02\%$ of weights for Q4\_K, $0.005\%$ for Q5\_K, $0.05\%$ for Q6\_K (whose single steps mostly fall below $\tau$), and $0.005\%$ for IQ4\_XS. Each floor lies above the natural boundary-crossing noise, the median per-tensor fraction of weights whose benign drift itself exceeds $\tau$ ($0.018\%$ for Q4\_K, at most $0.002\%$ for the others; distinct from the $0.22\%$ of weights touched at \emph{any} magnitude) (Fig.~\ref{fig:security}b). One budget step above each floor, TPR at FPR $1\%$ is $1.0$ in every seed. We report this as evidence that the distributions are separated, not as a defense we propose; a statistic exploiting an attack's spatial clustering may lower the floor (future work).

\paragraph{The floor depends on the environment pair} Benign differences vary by an order of magnitude across the Q4\_K environment pairs we have (Fig.~\ref{fig:security}c): $0.018\%$ of weights exceed $\tau$ between the arm64 builds, $0.034\%$ with an imatrix, and $0.13$--$0.19\%$ between the native x86 build and the reference; the publisher's 1.5B artifact (the pilot of \S\ref{sec:m1}) against our \texttt{b3772} rebuild sits at $0.13\%$, like the x86 pair. The per-tensor floor rises accordingly, from $0.02\%$ to $0.1\%$. This separation is robust to the tail threshold: repeating the analysis at $\tau\in\{0.05,0.2\}$ preserves the pair ordering and flags no honest artifact, with the model-level threshold scaling as expected ($0.41\%$, $0.19\%$, $0.08\%$ at $\tau{=}0.05,0.1,0.2$) and the max threshold, which does not involve $\tau$, unchanged at $2.36$. A verifier who does not know which environment produced an artifact must allow for all of the pairs it may face: calibrating a model-level rule on the five known-release pairs (flag a model whose pooled tail exceeds $0.19\%$, the largest benign value among them) raises no alarm on those pairs by construction, none on the held-out published artifact, and flags every tampered model from $0.2\%$ of weights up, in every pair and seed. Read the other way, this is a hiding capacity that grows with the range of environments a verifier must allow: in the 1.1B Q4\_K weights of the 1.5B model, up to about $2\times10^{6}$ uniformly placed one-step edits (just under $0.2\%$, for the five-pair calibration and an attacker on a low-noise pair) stay within what the statistic must attribute to benign environment variation, and more as that range widens. For a verifier who knows the environment pair but cannot replicate the producer's build, the capacity falls to about $2\times10^{5}$ (one who can replicate it simply compares hashes). Either way it is a quantitative case for $B{=}0$, under which a single edit changes the hash.

\paragraph{Sparse large edits} The tail statistic counts edits, so a handful of large ones escape it. A max statistic---the largest $|\Delta w|/\sigma$ in the model---does not: across the Q4\_K pairs, honest artifacts reach at most $2.36$, while moving a single weight's code to its farthest value already exceeds that in every seed and pair ($2.40$--$2.52$), and five such edits reach at least $2.57$. The two statistics are complementary---many small edits raise the tail count, a few large edits raise the maximum---so an attacker who stays below both must make few, small edits, which is the capacity bounded above.

\paragraph{Scope} These are uniform, untargeted code edits on one model family. We additionally evaluate three adaptive variants. The first perturbs per-block \emph{scales} (\S\ref{sec:bg}), which moves every weight in a block by sub-step amounts. When we multiply every Q4\_K super-scale by $1{+}\epsilon$ and sweep $\epsilon$, the pooled tail catches every perturbation in the sweep that is large enough to harm the model---at $\epsilon{=}0.02$ perplexity diverges ($>\!10^{6}$) and the pooled tail is $1.3\%$, well above the $0.19\%$ threshold---while $\epsilon{=}0.005$ evades all three statistics (per-tensor tail, pooled tail, max) at a perplexity cost of $+0.29\%$, and $\epsilon{=}0.01$ still evades the pooled and max statistics (the per-tensor statistic is borderline there) at $+0.98\%$---a few times to an order of magnitude above the benign environment drift ($+0.06\%$), but far below a damaging level (the sweep scores each perturbed artifact against the unperturbed one). In this sweep, the blind spot to a uniform scale perturbation is bounded and behaviorally minor. The second variant is adaptive \emph{placement}: an adversary who flips the codes whose one-step $|\Delta w|/\sigma$ is smallest hides up to the reservoir of such weights---$0.43\%$ of the Q4\_K weights here ($\approx\!4.7$M)---with the pooled tail near zero and the max statistic below $0.12$, evading all three statistics, whereas the same number of \emph{uniformly} placed flips is caught (max $2.4$--$7.1>2.36$ across the budgets tested). This is a genuine blind spot, but a bounded, low-leverage one: the hidden edits are by construction the lowest-magnitude ones (their total weight change is about $4$--$8\times$ smaller than uniform placement). The third, \emph{targeted} case (e.g.\ rounding-guided injection~\cite{chen2026qura}) is different in kind, and we treat it next. When $B{=}0$, all of these remain soundly detectable by exact re-derivation regardless (\S\ref{sec:protocol}).

\paragraph{A targeted attack defeats the statistic} The bounded blind spots above are for untargeted edits. A \emph{targeted} tamper is different in kind. Qwen2.5-1.5B ties its input and output embeddings, so a token's output logit is $\langle h,E_t\rangle$ for the final hidden state $h$; an adversary can raise a chosen token's logit by one-step edits to its embedding row, and---because that row is \texttt{Q6\_K}, whose one-step move is below $\tau$ for $95.5\%$ of elements---can restrict itself to sub-$\tau$ edits. Doing so flips the model's greedy next token to the chosen one while \emph{every} edit stays below $\tau$ (max $|\Delta w|/\sigma\le0.10$, no edit above $\tau$, pooled contribution ${\approx}10^{-6}$), evading all three statistics; we confirmed the flip on two probes by a full-precision forward pass with the tampered, dequantized embedding swapped in (e.g.\ ``the opposite of hot is'' $\to$ \emph{hot}, and a factual completion redirected to a rank-10 alternative). One embedding row's sub-$\tau$ budget covers a $\sim\!4$-logit margin here---enough to promote a plausible alternative to the top; a far target needs above-$\tau$ edits or several rows, and on a third probe (``two plus two equals'') the predicted margin was reached with sub-$\tau$ edits but the forward pass did not flip, so one row's budget is not always sufficient. The statistic therefore characterizes the benign and untargeted regimes but is \emph{not} a guarantee against a targeted cheating quantizer. This is why the protocol rests on $B{=}0$ exact re-derivation rather than statistical detection: when $B{=}0$, any edit---targeted or not, sub-$\tau$ or not---is a hash mismatch (\S\ref{sec:protocol}).

\paragraph{Threat boundary} The characterization above concerns a cheating quantizer ($Q\neq\Tcal$). An adversarial-fp16 attack such as Mind the Gap~\cite{egashira2025mindthegap} differs in kind: the publisher honestly quantizes a maliciously prepared fp16 ($Q=\Tcal$), so the signed $W$ is itself the attack and re-derivation reproduces the artifact exactly---it lies outside re-derivation by definition (A1), not by measurement, and neither our re-derivation analysis nor the statistic above addresses it.

\begin{figure*}[t]\centering
\includegraphics[width=\textwidth]{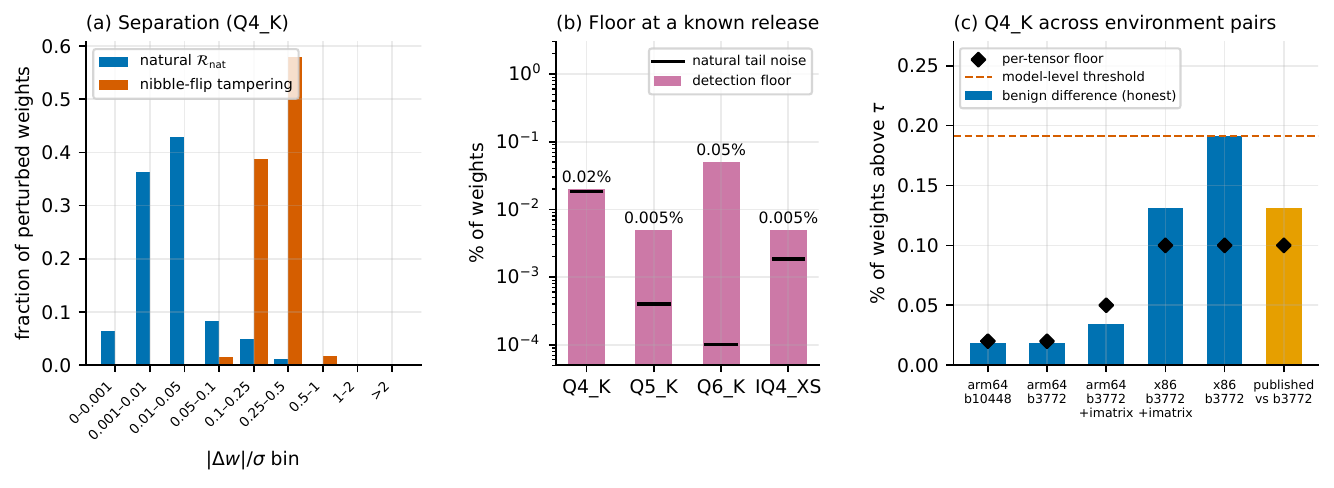}
\caption{Security characterization (Qwen2.5-1.5B-Instruct): (a) natural vs.\ nibble-flip perturbation distributions (Q4\_K); (b) per-tensor detection floor at a known release (arm64 contraction pair; tampered = honest rebuild plus uniformly placed one-step flips; AUC $\ge0.99$ in every seed) and the natural tail noise it must exceed; (c) Q4\_K across environment pairs: benign difference (fraction of weights above $\tau$), each pair's per-tensor floor, and the model-level threshold calibrated on the known-release pairs.}\label{fig:security}
\end{figure*}

\section{Verification Protocol}\label{sec:protocol}
The measurements above compose into a verification recommendation whose parts depend on one another in a specific order. The precondition is $B{=}0$ over the environments rebuilders use ($B_{\Ecal'}$, \S\ref{sec:formal}): when the honest output is unique, an honest independent rebuilder that recomputes $Q$ and compares hashes can soundly detect any deviation $Q\neq\Tcal(\Xbar;E)$---assuming the committed inputs and hash are sound (A2, A4)---so a cheating quantizer cannot hide an edit to the fp16$\to Q$ step (R1, Appendix~\ref{app:proofs}). Inputs it supplies itself ($W$, $C$, $M$) are bound by the manifest, not vetted (\S\ref{sec:bg}). $B{=}0$ is sufficient rather than necessary: if the plausible environments form a small, known set $\Ecal'$, a verifier can instead accept $Q$ when it lies in $\{\Tcal(\Xbar;E):E\in\Ecal'\}$. Without a recorded $E$ that set is not known: our \texttt{b3772} pilot (\S\ref{sec:m1}) enumerated three distinct honest outputs, including a native FMA-capable x86 build, and the published artifact was none of them. Recording $E$ makes membership checkable, and $B{=}0$ over the rebuilders' environments reduces it to a single hash comparison. When $B{>}0$, a nonzero diff no longer implies tampering (\S\ref{sec:m3}); verification can no longer rely on exact hash equality, and for the evaluated cheating-quantizer threat, detection is instead characterized by a detection floor that grows with the range of environments the verifier must allow ($0.02$--$0.2\%$ of Q4\_K weights, \S\ref{sec:m3}). Because $B{=}0$ over the tested platforms is attainable by one compiler flag (PR~\#25353) at negligible measured quality cost (\S\ref{sec:m2}), yet is not the \texttt{llama.cpp} default (the PR remains open as of September 2026), the first recommendation is simply to make deterministic quantization the default. This does not remove $E$ from the specification: the flag is itself an $E$ coordinate, and $B{=}0$ is an empirical property of the tested (architecture, compiler) pairs at the releases and types we tested (\texttt{b10448}; \texttt{b3772} and \texttt{b9222} for \texttt{Q4\_K\_M}) rather than a universal one: at \texttt{b9222} a \texttt{Q8\_0} tensor with all-zero blocks still stores architecture-dependent codes (\S\ref{sec:m1}), so the record must include the architecture or the kernel must canonicalize zero blocks. The manifest below therefore still records $E$---enabling the flag only makes that record cheap (the quantizer version plus one build flag).

Given $B{=}0$, the architecture is the reproducible-builds pattern adapted to quantization, differing from today's ecosystem in \emph{what is signed}. The transforming party---the quantizer, not only the original model signer---signs a manifest committing to the inputs it used, $(H(W),\theta,C,E,M)$, so that the claim ``$Q=\Tcal(\Xbar;E)$'' becomes checkable rather than asserted. Independent rebuilders recompute $Q$ from the manifest and publish signed match/mismatch results to an append-only transparency log (publicly readable and append-only, so a mismatch report cannot be quietly removed); end users verify the published results without performing the rebuild themselves. This inverts the current trust relation, in which a signature over $Q$ certifies its origin but not its faithfulness to $W$ (\S\ref{sec:bg}). The primitives here are standard---a signed manifest as in in-toto's link attestations~\cite{torresarias2019intoto}, an append-only transparency log as in Sigstore's Rekor~\cite{newman2022sigstore}, and cross-machine bit-exact re-derivation as ggufpacker already performs for the quants it packs~\cite{ggufpacker}. Our contribution is not this machinery but the measured specification of \emph{what it must bind} for a GGUF quantization to verify: the build's FP-contraction coordinate (\S\ref{sec:m2}), the per-tensor recipe (\S\ref{sec:m1}), and a hash of the exact imatrix, which \S\ref{sec:m1} shows an independent party could not recompute on the CPUs we tested. A publisher-side packager pins these coordinates implicitly for the quants it packs, but none of the signing or attestation specifications we examine (\S\ref{sec:bg}) requires them.

A tempting shortcut is to verify only a random sample of the artifact rather than re-derive all of it: the manifest would additionally commit to $Q$ via a Merkle root over its $N$ superblocks, letting a verifier \emph{open}---request and re-derive---individual units. Its cost, however, grows precisely as the adversary that matters---sparse tampering---becomes sparser. The bound below is the standard spot-checking argument of provable data possession~\cite{ateniese2007pdp}, which we apply to partial re-derivation.

\begin{proposition}[Sampling bound; standard, cf.~\cite{ateniese2007pdp}]\label{prop:sample}
A partial verifier re-deriving a uniform $k$-subset of $N$ units catches an $m$-unit tampering with $P_{\det}=1-\binom{N-m}{k}/\binom{N}{k}$, satisfying $1-e^{-km/N}\le P_{\det}\le km/N$. Attaining $P_{\det}\ge 1-\delta$ needs $k=\Theta(N/m)$.
\end{proposition}
\emph{Proof sketch.} Union bound on $\mathbb{E}|U\cap S|{=}km/N$ for the sampled $k$-subset $U$ and tampered set $S$ (necessity) and $\binom{N-m}{k}/\binom{N}{k}\le e^{-km/N}$ (sufficiency); Appendix~\ref{app:proofs}. For sparse tampering the required budget grows as $N/m$, approaching full re-derivation as $m\to1$, so spot-checking is weakest where detection matters most (and each opened unit also needs its fp16 inputs authenticated against $H(W)$, which a flat hash supports only by downloading $W$); we favor full re-derivation when reproducible, with attestation as a \emph{residual} mechanism only where re-derivation is unavailable.

\section{Discussion \& Limitations}\label{sec:disc}
Relative to ggufpacker and PR~\#25353, what is new is the ecosystem measurement and its decomposition, the localization of the failure to the quantization step (the fp16 re-derives exactly, the quantized weights do not), the two card-absent coordinates and the imatrix we could not recompute, which a manifest must therefore bind, the vocabulary that makes the verifier's requirement precise, and the security characterization of the $B{>}0$ regime. \textbf{Limitations:} models (Qwen 0.5B--14B, one Granite 3B), types (Q4\_K, Q5\_K, Q6\_K, IQ4\_XS, plus Q8\_0 as a determinism control; all 23 published types of one repository for re-derivation only), pipeline (\texttt{llama.cpp} CPU quantization). Out of scope: GPU imatrix generation $C$ and GPU-native pipelines (GPTQ/AWQ/bitsandbytes), which may require attestation rather than re-derivation; and adversarial-fp16 attacks, which no transformation-level check---re-derivation or attestation---detects, since the transformation is honest. The floors are measured for uniformly placed one-step flips (and farthest-code edits for the max statistic) against six environment pairs of one model family, plus uniform scale-perturbation, adaptive-placement, and targeted-embedding sweeps (\S\ref{sec:m3}). The targeted embedding attack defeats the statistic, which is why the protocol rests on $B{=}0$ exact re-derivation rather than detection. We do not evaluate adversarial-fp16 attacks; a preliminary run of the Mind-the-Gap code on a 0.5B model did not yield a working attack and is reported in the artifact only. The 12-card LLM audit is a sample, complemented by a pattern scan of all 198 cards (\S\ref{sec:m1}). \textbf{Recommendations:} \texttt{llama.cpp}/GGUF should record the quantizer version and build flags and merge PR \#25353; IETF/ML-BOM should bind $C$ and $E$.

\section{Related Work}\label{sec:rw}
\textbf{Reproducible builds.} The paradigm~\cite{lamb2022reproducible,reproduciblebuilds}, its tooling~\cite{diffoscope,reprotest}, and automated localization of non-reproducibility~\cite{ren2018reploc} establish an independently verifiable source-to-binary correspondence, reducing the trust placed in the build process. We transfer this lens to model quantization and find that metadata normalization cannot reach its non-determinism; the environment must then be made irrelevant or recorded~\cite{rbrecording}, and neither is current practice (\S\ref{sec:m1}, \S\ref{sec:m2}).
\textbf{Model signing \& attestation.} Sigstore Model Signing~\cite{sigstoremodeltransparency,openssfmodelsigning} authenticates the identity of the released files; provenance and lifecycle efforts---the IETF lifecycle-attestation draft~\cite{sharif2026ietfattest}, the CycloneDX ML-BOM transparency inventory~\cite{cyclonedxmlbom}, and attestation of LLM pipelines~\cite{tan2026attesting}---describe more of the pipeline, and some record deployment or toolchain context. What none of the specifications we examine \emph{requires} is the specific calibration artifact $C$, or execution-environment information (compiler, architecture, FP-contraction setting) in a form a verifier can re-create, both of which our bit-exact GGUF re-derivation protocol consumes (\S\ref{sec:bg}). Our manifest is an in-toto-style provenance statement~\cite{torresarias2019intoto} published to a transparency log as in Sigstore~\cite{newman2022sigstore}; SLSA provenance~\cite{slsa} has generic slots that could carry $C$ (resolved dependencies) and $E$ (build parameters), so what quantization needs is a profile that makes them mandatory and names their contents. Empirical studies of the Hugging Face ecosystem report missing provenance and reproducibility attributes as a general reuse challenge~\cite{jiang2023hfreuse}, and model-card completeness is uneven~\cite{liang2024modelcards}; our census quantifies one such attribute, the quantization environment, for the deployed GGUF slice. A concurrent registry study functionally tests 327 GGUF artifacts and finds silently defective official builds~\cite{patodiya2026broken}; it checks behavior where we check provenance by re-derivation, and the two are complementary. Statistical certification of a compressed model's quality~\cite{singh2026certifying} is likewise orthogonal to provenance. Attestable Builds~\cite{hugenroth2025attestable} obtain verifiable binaries from TEEs instead of rebuilds, the residual route we leave to attestation.
\textbf{Quantization-time attacks.} Mind the Gap~\cite{egashira2025mindthegap}, its predecessor on LLM quantization~\cite{egashira2024exploiting}, and, more generally, Qu-ANTI-zation and quantization backdoors against commercial toolkits~\cite{hong2021quantization,ma2023quantbackdoor} demonstrate \emph{quantization-conditioned} attacks: an adversarially prepared full-precision model behaves maliciously only after quantization (targeted misclassification or backdoor, across quantization schemes). Because the quantization step itself is faithful ($Q=\Tcal$), these fall outside our cheating-quantizer threat model---re-derivation reproduces them rather than flagging them---but they define the threat class that \S\ref{sec:m3} places outside re-derivation. Closer to our cheating quantizer, rounding-guided backdoor injection~\cite{chen2026qura} flips the rounding direction of selected weights during quantization, and bit-flip attacks show that a few targeted weight changes can suffice~\cite{coalson2024prisonbreak}; both are the natural attacks to place against the statistics of \S\ref{sec:m3}, whose targeted-embedding evaluation is a simple exact, gradient-free relative of the former. Defenses that re-round weights, such as EFRAP~\cite{li2024efrap} and QuantGuard~\cite{zheng2026quantguard}, would themselves depart from a canonical bit-exact $\Tcal$.
\textbf{Non-weight (metadata) backdoors.} A parallel surface is the non-weight component $M$: training-free backdoors carried in the chat template, with the weights left untouched~\cite{fogel2026chattemplate,wang2026badtemplate}. These motivate treating non-weight components such as chat templates as part of the signed manifest (\S\ref{sec:protocol}), though detecting them is orthogonal to re-derivation.
\textbf{Verifying ML computation.} zkML~\cite{chen2024zkml} and locality-sensitive commitments over activations~\cite{ong2025toploc} verify \emph{inference} computations, and deterministic-inference work removes floating-point non-determinism from the forward pass~\cite{he2025nondeterminism} or replays it bit-exactly as an audit signature~\cite{cankaya2026bitexact}. All target the inference that consumes a deployed model, not the \emph{transformation} $\Tcal$ that produces it; our focus is the latter, verified by re-derivation.

\section{Conclusion}\label{sec:conc}
The ML supply chain authenticates identity but deploys transformations that are insufficiently specified to support independent bit-exact verification; in our popularity-sampled slice the recorded-inputs \ivr{} is $0/198$, and re-derivation by environment search succeeded, in our trial, for recent releases but not for an older one. For the studied \texttt{llama.cpp} quantization pipeline, the environment gap is closable for newly produced artifacts on the tested platforms---by enabling deterministic quantization, at negligible measured quality cost, and recording the release and build---while the remaining source gaps require publishers to release their fp16 and imatrix, together with the per-tensor recipe, since an independent party could not recompute the imatrix in our attempt; adoption carries engineering cost. Until it closes, benign environment variation leaves room to hide: up to about $0.2\%$ of uniformly placed one-step weight edits (some $2\times10^{6}$ at 1.5B, and more as the admissible environment range widens) stay within what the evaluated statistic must attribute to it, an adaptively placed adversary hides a further low-magnitude reservoir, and a targeted edit evades the statistic outright. We provide a re-runnable audit and tools to track the \ivr{}.

\section*{Ethical Considerations}
This is a measurement and characterization study; it introduces no new attack class. Our audit reads only public information from public Hugging Face repositories (file inventories, model cards, and GGUF headers fetched by range requests); we collect no personal data and report aggregate rates rather than judgements of specific publishers---the repositories we name (\S\ref{sec:m1}) are best cases whose publisher discloses more than most. The cheating-quantizer analysis and detector are defensive---they separate benign non-determinism from tampering to aid verification. We ran one previously published attack procedure (Mind the Gap) with its authors' public code on a small model, which did not yield a working backdoor in our run, and we redistribute no attacked model; the synthetic nibble-flip constructions and the targeted-embedding check (which flips a next-token prediction on a neutral probe) are diagnostic detector evaluations, not deployable backdoors. Our central recommendation---recording the execution environment and enabling deterministic quantization, as an already-proposed upstream change (PR~\#25353) does---is a defensive improvement to the supply chain.

\section*{Open Science}
All artifacts needed to evaluate our claims are available at \url{https://github.com/Rlonglong/gguf-verifiability-artifact}.
They comprise the frozen 198-repository list with the scored census snapshot and the re-runnable scoring harness (\texttt{census.py}); the model-card audit verdicts and the classified pattern scan of all 198 cards; the GGUF header-probe, byte-region diff, tensor-data hash, and detector tools, with the detector evaluation over environment pairs; the cross-architecture, pilot and trial re-derivation (\texttt{b3772}, \texttt{b9222}, \texttt{b10603}), imatrix-regeneration, conversion, and perplexity drivers; the scale-perturbation, adaptive-placement, and targeted-embedding sweeps; and the raw per-figure data (CSV/JSON) with the scripts that render every figure and reported number. The inputs are public: Hugging Face model artifacts, pinned open-source \texttt{llama.cpp} builds (\texttt{b3772}, \texttt{b9222}, \texttt{b10448}, \texttt{b10603}), and, for the Mind-the-Gap workload, its authors' public code; we reference these by identifier rather than redistribute them. A Zenodo deposit will accompany the published version.

\section*{LLM usage considerations}
LLMs were used for editorial purposes in this manuscript, and all outputs were inspected by the authors to ensure accuracy and originality; the authors take full responsibility for every statement, number, and line of code in the paper and artifact. We also disclose two further uses. \emph{Methodology:} the model-card audit of \S\ref{sec:m1} used LLM agents (one Claude agent per card) instructed to refute our scorer's label. The \ivr{} itself rests on deterministic metadata checks, which the agents could only confirm or contest; because the agent model is closed and its outputs may not be exactly reproducible, we complement the sample with a deterministic pattern scan of all 198 cards whose every hit we classified manually, and we release the per-card verdicts with their rationales so the conclusion can be checked without an LLM. \emph{Engineering:} an LLM coding assistant helped write experiment and analysis scripts; every reported number comes from executed code and released raw data and was checked against that data by the authors, and every reference was verified against its source. \emph{Footprint:} we train no models beyond the few hundred fine-tuning steps of the preliminary Mind-the-Gap run; our experiments are CPU quantization, a handful of perplexity evaluations, and a few GPU forward passes for the targeted-embedding check, on open models of at most 14B parameters, on one Apple M4 machine and one x86-64 workstation, and the census fetches only GGUF headers (8--16\,MB per repository) instead of full model files.

\bibliographystyle{ieeetr}
\bibliography{refs}

\appendices
\section{Proofs}\label{app:proofs}
\textbf{Notation (recap).} Split the non-environment inputs from the environment: $\Xbar=(W,\theta,C,M)$, full input $X=(\Xbar,E)$ with $E\in\Ecal$; the implementation family is $\{\Tcal(\Xbar;E):E\in\Ecal\}$, the admissible-output set $\Acal(\Xbar)=\{\Tcal(\Xbar;E):E\in\Ecal\}$, and $B(\Xbar)=\sup_{Q_1,Q_2\in\Acal(\Xbar)} d(Q_1,Q_2)$ with $d$ the discrete metric on serialized artifacts (\S\ref{sec:formal}).

\subsection{Proof of Observation~\ref{prop:e} (necessity of $E$)}\label{app:proof-e}
\begin{proof}
Fix a single instance $\Xbar$: $W=$ Qwen2.5-1.5B-Instruct fp16 (a fixed GGUF with $H(W)$ identical on both platforms), $\theta=$\texttt{Q4\_K\_M}, $C=\varnothing$ (no imatrix), and $M$ the model's metadata. Take two environments $E_1=(\text{x86-64, gcc, default FP contraction})$ and $E_2=(\text{arm64, clang, default FP contraction})$, with an identical pinned build configuration (CPU-only, static, Release) on both; thread counts differed (8 vs 6) and are separately observed not to affect the output (every configuration is bit-deterministic over 10 repeated runs, and the 8- and 6-thread \texttt{-ffp-contract=off} outputs are byte-identical). Building \texttt{llama.cpp} at the pinned release build \texttt{b10448} and quantizing on each platform (\S\ref{sec:m2}) yields $\Tcal(\Xbar;E_1)\neq\Tcal(\Xbar;E_2)$: the two outputs byte-differ ($H(\Tcal(\Xbar;E_1))\neq H(\Tcal(\Xbar;E_2))$), with the divergence confined to the k-quant weight tensors (the metadata/KV region is byte-identical) and originating in the FMA-sensitive scale-search loop of \texttt{ggml-quants.c}. Hence $|\Acal(\Xbar)|\ge 2$ and $B(\Xbar)>0$: the specification that pins only $\Xbar$ and leaves $E$ free is incomplete, so any complete specification for this transformation family must additionally constrain $E$; thus $E\in X^*$.
\end{proof}
\paragraph{Remark (a sufficient completion; not a converse)} Adding the single coordinate $E_{\text{fpcontract}}$ (compiling \texttt{ggml-quants.c} with \texttt{-ffp-contract=off}) makes the output bit-identical across the two environments above, i.e.\ across both tested (architecture, compiler) pairs (\S\ref{sec:m2}: $4/4$ bit-identical for the four types tested there, including the scale-search type \texttt{IQ4\_XS}); thus $\Xbar\cup\{E_{\text{fpcontract}}\}$ is a \emph{sufficient} completion for the tested environments. This is separate from the observation itself, which establishes only that $E$ is necessary. We do \emph{not} claim (i) that $E_{\text{fpcontract}}$ suffices over \emph{all} environments (untested compilers/GPUs may add coordinates), nor (ii) that \emph{every} transformation model requires $E$---a family fixing a single canonical implementation would relocate the requirement into that implementation's specification rather than eliminate it. Characterizing $X^*$ in full (GPU imatrix, other pipelines) is future work.
\paragraph{Consequence (specification gap)} The deployed specifications we examine---Model Signing ($H(W)$ only), the IETF lifecycle-attestation draft, CycloneDX ML-BOM---do not require the execution-environment coordinates (compiler, target architecture, FP-contraction setting) that determine the output; the IETF draft binds deployment/TEE context and a quantization-tool hash, which identifies a binary but does not record the coordinates needed to re-create an equivalent one on another platform. They are therefore incomplete for independent re-derivation over the studied family: under their required fields, a verifier without the producer's binary and platform cannot determine which $Q$ to expect. This underpins the execution-environment gap measured in \S\ref{sec:m1} (the build coordinates of $E$ absent from the GGUF metadata and the model cards of all 198 audited artifacts). We state this as \emph{the necessity of $E$ explains why the universal gap blocks re-derivation of scale-search artifacts}, not as \emph{a $0$ \ivr{} is a formal consequence of missing $E$ alone} (a $0$ \ivr{} can also arise from missing $C$, $W$, or $\theta$).

\subsection{Proof of Proposition~\ref{prop:sample} (sampling bound; standard)}\label{app:proof-sample}
\textbf{Setup.} Partition the artifact into $N$ units (superblocks). The honest transformation fixes each unit's value. A tampering adversary alters a set $S$, $|S|=m$, with $Q\neq\Tcal(\Xbar;E)$ exactly on $S$. A partial verifier re-derives a uniformly random subset $U\subseteq[N]$, $|U|=k$, and reports tampering iff $U\cap S\neq\varnothing$.
\begin{proof}
\emph{Exact form.} $\Pr[U\cap S=\varnothing]$ is the probability a uniform $k$-subset avoids the fixed $m$-set $S$, $=\binom{N-m}{k}/\binom{N}{k}$; hence $P_{\det}=1-\binom{N-m}{k}/\binom{N}{k}$.

\emph{Upper bound (necessity).} With $Z_i=\mathbf{1}[i\in U\cap S]$ and $\mathbb{E}|U\cap S|=k\cdot(m/N)$, the union bound gives $P_{\det}=\Pr[|U\cap S|\ge 1]\le \mathbb{E}|U\cap S|=km/N$. Thus $P_{\det}\ge 1-\delta \Rightarrow k\ge(1-\delta)N/m$.

\emph{Lower bound (sufficiency).} Writing the ratio as a telescoping product and using $(N-m-i)/(N-i)\le 1-m/N$ for each $0\le i<k$,
\[
\frac{\binom{N-m}{k}}{\binom{N}{k}}=\prod_{i=0}^{k-1}\frac{N-m-i}{N-i}\le\Bigl(1-\tfrac{m}{N}\Bigr)^{k}\le e^{-km/N},
\]
so $P_{\det}\ge 1-e^{-km/N}$, and $e^{-km/N}\le\delta \Leftrightarrow k\ge (N/m)\ln(1/\delta)$ suffices. Combining the two bounds, $k=\Theta(N/m)$.
\end{proof}
\paragraph{Corollary (spot-checking degrades with sparsity)} For a 7B model, $N\approx 2.7\times10^{7}$ superblocks. Against sparse tampering $m=10^{2}$, $\delta=0.05$: the necessary budget is $k\ge 0.95\,N/m\approx 2.6\times10^{5}$, and $k=\ln(20)\,N/m\approx 8.1\times10^{5}$ suffices. Each opened unit needs a Merkle path of $\lceil\log_2 N\rceil\approx 25$ sibling hashes ($\sim$32\,B each), i.e.\ $\sim$800\,B, for a proof of $\approx 0.2$--$0.65$\,GB as independent paths (a multiproof sharing upper-level nodes needs about $k\log_2(N/k)$ hashes, a few times fewer). The proposition establishes only that the spot-check budget scales as $\Theta(N/m)$; from this we draw two \emph{engineering} conclusions (not theorem consequences): for sufficiently sparse tampering, spot-checking is impractical---its proof grows as $N/m$ and approaches full re-derivation as $m\to1$ (and authenticating the opened units' fp16 inputs against a flat $H(W)$ requires downloading $W$, $\approx$14\,GB, anyway); and when full re-derivation is feasible ($B{=}0$) it avoids this budget entirely, with attestation a residual mechanism where re-derivation is unavailable.

\subsection{Remark}
\textbf{R1.} For a transformation instance with $B_{\Ecal'}(\Xbar){=}0$ over the rebuilders' environments $\Ecal'$, the admissible output is unique; under A2 and A4, an honest rebuilder recomputing the transformation detects any $Q\neq\Tcal(\Xbar;E)$ except with negligible hash-collision probability.

\section{Assumptions and Security-Claim Boundary}\label{app:assume}
Our claims hold under the following assumptions. A1--A4 govern the primary contribution (re-derivation and the structural characterization); A5--A7 are needed \emph{only} for the residual attestation fallback (\S\ref{sec:protocol}), which we treat as future work.
\begin{description}
\item[A1 (Boundary).] The source signature $\mathsf{sig}_W$ is unforgeable and the signed source $W$ is itself trustworthy. We do \emph{not} address upstream poisoning: if $W$ is malicious (an adversarial-fp16 model, Mind the Gap), an honest quantization of it is faithful, and neither re-derivation nor our detector flags it (\S\ref{sec:m3}).
\item[A2.] The hash $H$ is collision-resistant; underlies both the manifest commitments and the equality $H(Q')=H(Q)$.
\item[A3 (modeling scope, not a free assumption).] We model the studied transformation's inputs as $X=(\Xbar,E)$ with $\Xbar=(W,\theta,C,M)$---the candidate coordinates that may influence its output. This is the model we \emph{test}, not something we take for granted: Observation~\ref{prop:e} and \S\ref{sec:m2} establish that $E$ must be constrained, \S\ref{sec:m1} measures that deployed specifications omit it, and characterizing $X^*$ in full is open.
\item[A4.] At least one rebuilder executes the specified transformation faithfully on the committed inputs. When $B_{\Ecal'}{=}0$, one honest rebuilder that reproduces the canonical output already detects a deviation; using several independent rebuilders is an operational robustness measure, not part of the soundness argument.
\item[A5--A7 (residual only).] TEE hardware isolation is sound (A5); attestation is unforgeable and its measurement faithfully reflects the executed binary and configuration (A6); the approved quantizer binary has been deemed acceptable out-of-band (A7)---a TEE proves an approved quantizer \emph{ran on committed inputs}, not that its algorithm is correct. These delimit the residual attestation architecture only; we prove no TEE security theorem.
\end{description}
\paragraph{Security-claim boundary, by regime} Against a \emph{cheating quantizer} ($W$ benign, $Q\neq\Tcal(\Xbar;E)$) we state no single unconditional ``detectable'': (i) when $B_{\Ecal'}{=}0$ over the rebuilders' environments (the target state), any $Q\neq\Tcal(\Xbar;E)$ is \emph{soundly detectable} by exact re-derivation; (ii) when $B_{\Ecal'}{>}0$ (the current ecosystem), re-derivation alone cannot decide; for uniformly placed nibble flips we give the \emph{measured detection floor} of \S\ref{sec:m3}, which ranges from $0.02\%$ to $0.2\%$ of Q4\_K weights over the environment pairs we measured; a uniform scale perturbation is caught before it is harmful, adaptive placement hides a bounded low-magnitude reservoir, and a targeted embedding edit evades the statistic entirely (\S\ref{sec:m3}), so the statistic is a characterization, not a guarantee; (iii) a \emph{malicious} $W$ (adversarial-fp16, $Q=\Tcal(\Xbar;E)$) is \emph{not} detectable by re-derivation (A1), since the quantization is honest and re-derivation reproduces it. We deliberately establish how far the no-trusted-hardware approach reaches before invoking attestation.

\end{document}